\documentclass[11pt,reqno]{amsart}
\usepackage{mathrsfs,amsmath,amsxtra,amssymb,amsthm,amsfonts}
\usepackage{mathtools}

\DeclarePairedDelimiter\floor{\lfloor}{\rfloor}

\renewcommand{\for}{\begin{eqnarray*}}
\newcommand{\mel}{\end{eqnarray*}}
\def\fr{\begin{align*}}

\newcommand{\ten}{\otimes}

\newcommand{\pl}{\hspace{.1cm}}

\newcommand{\si}{\sigma}

\newcommand{\la}{\lambda}

\newcommand{\A}{{\mathcal A}}
\newcommand{\B}{{\mathcal B}}

\newcommand{\R}{{\mathcal R}}

\newcommand{\norm}[2]{\parallel \! #1 \! \parallel_{#2}}

\newtheorem{lemma}{Lemma}[section]
\newtheorem{prop}[lemma]{Proposition}
\newtheorem{theorem}[lemma]{Theorem}
\newtheorem{cor}[lemma]{Corollary}

\newtheorem{rem}[lemma]{Remark}
\newtheorem{definition}[lemma]{Definition}

\newcommand{\re}{\begin{rem}\rm}
\newcommand{\mar}{\end{rem}}

\newcommand{\bra}[1]{\langle{#1}|}
\newcommand{\ket}[1]{|{#1}\rangle}
\newcommand{\ketbra}[1]{|{#1}\rangle\langle{#1}|}
\newcommand{\qd}{\end{proof}\vspace{0.5ex}}
\newcommand{\prf}{\begin{proof}[\bf Proof:]}

\newcommand{\xspace}{\hbox{\kern-2.5pt}}

\newtheorem*{theorem*}{Theorem}

\allowdisplaybreaks

\usepackage{verbatim}
\usepackage{amsmath}
\usepackage{amsthm}
\usepackage{listings}
\usepackage{graphicx}
\usepackage{float}
\usepackage{amsfonts}
\usepackage{latexsym}
\usepackage{color}
\usepackage{url}
\usepackage{hyperref}
\usepackage{setspace}
\usepackage[normalem]{ulem}
\usepackage{braket}
\usepackage[toc,page]{appendix}
\usepackage{cancel}

\title{Heralded Channel Holevo Superadditivity Bounds from Entanglement Monogamy}

\author{L. Gao}
\address{Department of Mathematics\\
University of Illinois, Urbana, IL 61801, USA} \email[Li Gao]{ligao3@illinois.edu}

\author[M. Junge]{M. Junge$^*$}\thanks{$^*$ Partially supported by NSF-DMS 1501103}
\address{Department of Mathematics\\
University of Illinois, Urbana, IL 61801, USA} \email[Marius
Junge]{mjunge@illinois.edu}

\author[N. LaRacuente]{N. LaRacuente$^{\dag}$}\thanks{$^{\dag}$ This material is based upon work supported by NSF Graduate Research Fellowship Program DGE-1144245}
\address{Department of Physics\\
University of Illinois, Urbana, IL 61801, USA} \email[Nicholas LaRacuente]{laracue2@illinois.edu}

\begin{document}

    \newenvironment{remark}[1][Remark]{\begin{trivlist}
    \item[\hskip \labelsep {\bfseries #1}]}{\end{trivlist}}
    \newenvironment{conjecture}[1][Conjecture]{\begin{trivlist}
    \item[\hskip \labelsep {\bfseries #1}]}{\end{trivlist}}
    \newenvironment{corollary}[1][Corollary]{\begin{trivlist}
    \item[\hskip \labelsep {\bfseries #1}]}{\end{trivlist}}

	\newcommand{\partdiff}[2]{\frac{\partial #1}{\partial #2}}
	
	\newcommand{\raising}[2]{\hat{#1}^{\dag #2}}
	\newcommand{\dual}[2]{#1^{* #2}}
	\newcommand{\fock}{\mathcal{F}}
	
	\newcommand{\hilbert}{\mathcal{H}}
	\newcommand{\banach}{l}
	\newcommand{\bigphi}{\boldsymbol{\Phi}}
	\newcommand{\bigpsi}{\boldsymbol{\Psi}}
	\newcommand{\sep}{\text{SEP}}
	\newcommand{\one}{1}
	\newcommand{\expect}{\mathbb{E}}
	\makeatletter
	\newcommand{\vast}{\bBigg@{4}}
	\newcommand{\Vast}{\bBigg@{5}}
	\makeatother
	\newcommand{\flambda}{\underline{\lambda}}

	\newcommand{\canchan}{\Theta}
	\newcommand{\canB}{C}

\begin{abstract}
We show that for a particular class of quantum channels, which we call heralded channels, monogamy of squashed entanglement limits the superadditivity of Holevo capacity. Heralded channels provide a means to understand the quantum erasure channel composed with an arbitrary other quantum channel, as well as common situations in experimental quantum information that involve frequent loss of qubits or failure of trials. We also show how entanglement monogamy applies to non-classicality in quantum games, and we consider how faithful, monogamous entanglement measures may bound other entanglement-dependent quantities in many-party scenarios.
\end{abstract}

\maketitle

\section{Introduction}
\label{sec:intro}
The aim of Quantum Shannon theory is to adapt the idea of  Shannon's  seminal 1948 paper \cite{shannon_mathematical_2001} to the realm of quantum mechanics. Shannon's program has been applied to different notions of capacity for quantum channels (see e.g. \cite{wilde}), including transmission of qubits, classical information transmission with quantum encoding, and transmission with physically ensured security. A challenge and opportunity of Quantum Shannon theory lies in the phenomenon of \textit{superadditivity}, in which the capacity to transmit information over a collection of quantum channels may exceed the sum of their individual capacities \cite{hastings_superadditivity_2009,smith_quantum_2008,cubitt_unbounded_2015,li_private_2009}. For classical channels, the Shannon capacity is additive, allowing a \textit{single-letter} entropy expression -- a formula defined for one use of the channel. In contrast, so far there is no such single-letter expression for the classical capacity of a quantum channel, and an open question as to the general tractability of capacity calculations \cite{shor_quantum_2009}. By the HSW (Holevo, Schumacher \& Westmoreland) theorem \cite{holevo_capacity_1998,schumacher_sending_1997}, the classical capacity $C(\Phi)$ of a quantum channel $\Phi$ is given by \textit{regularization} of \textit{Holevo information} $\chi(\Phi)$, which maximizes the entropy expression for $m$ copies of the channel at once, taking $m$ to infinity.
 
In many laboratory settings, we expect that not every run of a protocol will succeed. This is a common paradigm in linear optics: photon sources \cite{tomamichel_tight_2012}, protocols \cite{pisenti_distinguishability_2011, calsamiglia_maximum_2014} and gates \cite{knill_scheme_2001} all have a probability or amplitude of failing in such a way that we might only discover after a measurement, or at the end of the experiment. Many of these situations can be described by \textit{quantum erasure channels}, which have some probability of outputting a flagged error state. Sometimes the probabilistic failure may be composed with successful but imperfect transmission. Keeping with the experimental and photonics literature, we call these \textit{heralded channels} \cite{northup_quantum_2014, hofmann_heralded_2012, bernien_heralded_2013, usmani_heralded_2012, maunz_heralded_2009}. These are a special case of flagged quantum channels as discussed in \cite{cubitt_unbounded_2015, elkouss_nonconvexity_2016, fuchs_nonorthogonal_1997}, in which the classical output is alerted to the outcome of some non-deterministic channel selection process. A flagged channel randomly selects among a fixed ensemble of quantum channels to apply to its input states, and reports which was applied along with the output state. Some cases of these channels are mentioned in \cite{marvian_symmetry_2012,brandao_when_2012,smith_quantum_2008}. Our notion of the flagged channel differs from that of the compound channels discussed in \cite{bjelakovic_quantum_2008, bjelakovic_classical_2007,  bjelakovic_arbitrarily_2013}, as they consider an unknown channel taken from a known set that is applied consistently, while we consider a new channel that draws randomly from a set of known channels at each use, but which reveals the result of this random drawing at the output.
\begin{figure}[H]
\centering
\includegraphics[width=.6\textwidth]{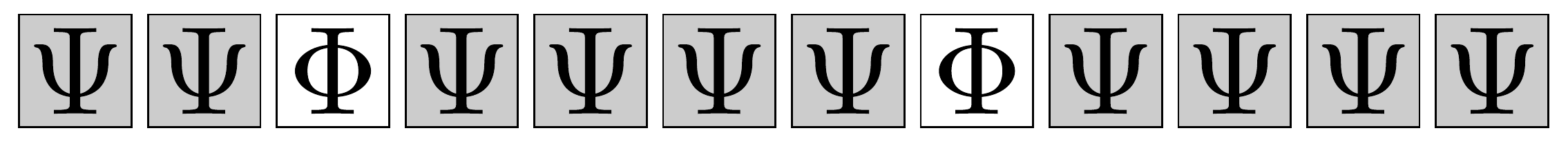}
\caption{Diagram of outcomes of a flagged channel that sometimes applies $\Phi$ and other times applies $\Psi$. We do not know for each turn which will be applied at preparation time, only at the output.}
\label{fig:her1}
\end{figure}
In these practical situations, the potential contribution of superadditivity is limited. Superadditivity relies on entangling input states to send over multiple uses of the channel. Entanglement as a quantum phenomenon is monogamous \cite{coffman_2000_distributed, monogamy, osborne_2006_general}. Thus unpredictably assigning nearly additive channels to many of the inputs dilutes the available entanglement for superadditivity of the desired ones. For example, even if one could efficiently prepare a large number of entangled photons at the input to a set of optic quantum channels, losing many of these photons means the entanglement after transmission is small. In this situation, an entangled scheme has limited or no advantage over a probably simpler non-entangled input, because entanglement monogamy limits the average entanglement between potentially superadditive outputs. Based on above considerations, we formulate the following conjecture:
\begin{conjecture}
\textit{A flagged channel that with high probability applies a strongly additive channel is nearly additive.}
\label{conj:hadd}
\end{conjecture}

Here the phrase ``strongly additive" means that it has no superadditive phenomenon when paired with any other channel. We quantify this conjecture for classical capacity, using the monogamy of squashed entanglement and its faithfulness with the trace distance \cite{li_squashed_2014}. Our results suggest one circumstance in which we can avoid the full difficulty of regularization. Surprisingly, post-selecting on success may not produce an exactly equivalent situation to having a channel without the failure probability, as it limits the phenomenon of superadditivity. We also show that with limited \textit{entangled blocksize}, the number of inputs that are not separable in a quantum code, a wide variety of quantities become nearly additive for heralded channels.

In addition to the experimental motivations for understanding heralding, we derive a new technique of fundamental interest in estimating quantum channel capacity and other entropy expressions. Instead of relying on the details of the channel, we show that entanglement monogamy limits the contribution of superadditivity when considering a small subset of outputs randomly chosen from a larger set. Going beyond the case of capacity and based on the faithfulness of squashed entanglement \cite{christandl_squashed_2004}, any entanglement-dependent quantity that is faithful in the trace distance should also be faithful in the squashed entanglement, hence admitting a bound from entanglement monogamy.
\subsection{Main contributions}
Let $\Phi$ be a quantum channel. That is, a completely positive trace preserving map which sends density operators (positive and trace one) of one Hilbert space $H_{A}$ to another $H_B$. The HSW theorem \cite{holevo_capacity_1998,schumacher_sending_1997} proves that the classical capacity of a quantum channel $C(\Phi)$ is given by \textit{regularization} of its \emph{Holevo information} $\chi(\Phi)$ as follows,
\begin{equation*}
\chi(\Phi) := \sup_{\rho^{XA}} I(X;B)_{\si} = \sup_{\rho^{XA}} S(B)_\si - \sum_x p(x) S(B)_{\si_x} \pl , \pl C(\Phi) = \frac{1}{m}\limsup_{m \rightarrow \infty} \chi(\Phi^{\otimes m}).
\label{eq:classcap}
\end{equation*}
where $S(\si)=-tr(\si\log \si)$ is the von Neumann entropy and $I(X;B)= S(X)+S(B)-S(XB)$ is the mutual information. The supremum runs over all classical-quantum inputs \[\rho^{XA}={\sum_{x}p(x)\ket{x}\bra{x}\ten \rho_x^{A}}\] and $\si=\Phi(\rho),\si_x=\Phi(\rho_x)$ are the outputs.

We define for $\la\in [0,1]$ the generalized quantum erasure channel with $1-\la$ erasure probability,
\begin{equation}
Z_\lambda(\Phi)(\rho) = \lambda \ket{0}\bra{0}^Y \otimes \Phi(\rho) + (1 - \lambda) \ket{1}\bra{1}^Y \otimes \si=\left[\begin{array}{cc}\lambda\Phi(\rho) & \\ & (1 - \lambda) \si\pl.
\end{array}\right]
\label{eq:eraserin}
\end{equation}
where $Y$ is the heralding/flagging classical output, and $\si$ is some fixed state. Using monogamy and faithfulness of squashed entanglement, we have the following estimate of its classical capacity:
\begin{theorem}
\label{thm:introer}
For any quantum channel $\Phi$ with output dimension at most $d$,
\begin{equation}
C(Z_\lambda(\Phi)) \leq \lambda  \chi(\Phi) + O( d(\log d)^{\frac{1}{4}} \la^{\frac{5}{4}} \log \lambda ) \pl.
\end{equation}
\end{theorem}
This result is a consequence of the following setting. Let $\Phi_1,\cdots,\Phi_n$ be a family of quantum channels. Given $k\le n$, we consider a \emph{heralded channel} with fixed success number $Z_k(\Phi_1,\cdots,\Phi_n)$ that randomly selects $k$ positions from $\{1,\cdots,n\}$, applying $\Phi$ channels to those positions and trivial channels $\Theta$ with fixed output to the rest. In addition, it provides an extra classical signal describing the chosen $k$ positions. Denote the integer set as $[n]=\{1,\cdots,n\}$. Mathematically, for an input state $\rho$ this is
\begin{equation}\label{A}
Z_{k} (\Phi_1, ..., \Phi_n)(\rho) =  \frac{1}{\binom{n}{k}} \sum_{R \subseteq [n], |R| = k} (\Phi^{R} \otimes \canchan^{R^c}) (\rho) \otimes \ketbra{R}^Y\pl,
\end{equation}
where the sum is for all $k$-subsets of positions $R\subseteq \{1,\cdots,n\}$, $R^c$ is the complement of $R$, and $\ketbra{R}^Y$ is the classical heralding signal. We use the notation $\Phi^{R} \otimes \canchan^{R^c}$ for the tensor channel which has $\Phi$ for the positions of $R$ and $\canchan$ in the complement $R^c$, and where the $\ketbra{R}^Y$ term is the direct output to auxiliary subsystem $Y$ (independent of input). This extra, classical output creates a block diagonal output structure, generalizing that of the erasure channel \eqref{eq:eraserin}.
\begin{theorem}
\label{thm:introher}
Let $\Phi_1,\cdots,\Phi_n$ be a family of channel with the dimension of the output system of each $\Phi$ is at most $d$. Then for $\lambda=\frac{k}{n}<1$,
\begin{equation}
\label{eq:intro1}
C(Z_k(\Phi_1, ..., \Phi_n)) \leq \lambda \sum_{j=1}^{n} \chi(\Phi_j)
	+ O \big ( k d(\log d)^{\frac{1}{4}} \la^{\frac{5}{4}} \log \lambda \big )\pl.
\end{equation}
\end{theorem}
The first term is single-letter, which is the average Holevo information of non-trivial channels according to the success probability $\lambda$. The second term bounds possible additivity violation. Note that using convexity of Holevo information in the input state \cite{gao_capacity_2016}, one may obtain linear correction term $O(\la)$. Here monogamy of entanglement narrows the additivity violation decay when the success probability $\la$ goes to $0$, without relying on the detailed form of $\Phi$'s.  The correction term decays faster than linearly, making it stronger than the linear correction term from convexity, as well as diamond norm comparison to channels with additive $\Phi_j$ or zero success probability \cite{leditzky_approaches_2017}.

Going beyond capacity, we construct a form of quantum game in which Alice prepares a joint state via free communication with many ``B" players simultaneously, though the dimension of her system is fixed and does not depend on the number of B players. After this initial state is prepared, the referee randomly selects one B player, and Alice plays a bipartite no-communication game with that B player. When there are many B players, monogamy of entanglement limits the extent to which the average value of these games with a quantum pre-shared state may exceed that with a classical pre-shared state. See \ref{sec:games} for the details of this setup.

We organize this work as follows: Section \ref{sec:prelim} introduces some notation and reviews channel capacity, squashed entanglement and continuity of entropies. In Section \ref{sec:heravg}, we explain our main technique, which we call ``heralded averaging." Section \ref{sec:it} gives the proof of Theorem \ref{thm:introher}. We use this in Section \ref{sec:related} to show Theorem \ref{thm:introer}. Section \ref{sec:other} discusses other applications of our principles to other capacities and quantum games. We end with discussion and conclusions in section \ref{sec:conclusion}.

\section{Preliminary}
\label{sec:prelim}
\subsection{Channel and capacity}

\label{sec:channeldef}
Let $\{\Phi_1, \cdots , \Phi_n\}$ and $\{\Psi_1, ... , \Psi_n\}$ be two classes of quantum channels. For $k<n$, we define the \emph{flagged switch channel} as follows,
\begin{align*}
Z_k(\Phi_1,\cdots, \Phi_n; \Psi_1,\cdots, \Psi_n) (\rho) =
	 \frac{1}{\binom{n}{k}} \sum_{R\subset [n],|R| = k} \Big (\Phi^{R} \otimes \Psi^{R^c}\Big ) (\rho) \otimes \ketbra{R}^{Y}\pl.
\end{align*}
This channel is a sum of tensor products of ``$\Phi$" and ``$\Psi$" channels, randomly selecting $k$ positions at which to apply the corresponding ``$\Phi$" channel, and applying $\Psi$ channels at all others. Here $R$ indexes the classical output $\ketbra{R}^{Y}$, representing which subset of the outputs that were $\Phi$ channels. We assume that the flagging process produces such a classical signal $Y$, such as by detection of a secondary photon at the output of a source \cite{tomamichel_tight_2012}, or the outcome of an attempted quantum experiment. That is by the time it reaches the output, the flagging signal has already been converted into a classical register $Y$ of value $R$, the states of which we denote $\ketbra{R}$. Note that we do not switch the positions of input channels, only choosing between alternatives at each position. If all $\Psi_j$ are erasure operations which have a fixed output $\Psi_j(\rho)=\Theta(\rho)= \si$ for some state $\si$, we call
\[Z_k(\Phi_1,\cdots, \Phi_n) := Z_k(\Phi_1,\cdots, \Phi_n ;\Theta,\cdots, \Theta)\]
the \emph{heralded channel}. This is because with a fixed output, any information transmission fails. To shorten our notation, in the following we will write
\[Z_k(\bigphi;\bigpsi):=Z_k(\Phi_1,\cdots, \Phi_n; \Psi_1,\cdots, \Psi_n)\pl, \pl Z_k(\bigphi):=Z_k(\Phi_1,\cdots, \Phi_n)\pl,\]
where the classes of channel $\bigphi=\{\Phi_1,\cdots,\Phi_n\}$ and $\bigpsi=\{\Psi_1,\cdots,\Psi_n\}$ are clear. Some cases of these channels are mentioned in \cite{marvian_symmetry_2012,brandao_when_2012,smith_quantum_2008}.

Our notion of the flagged channel differs from that of the compound channels discussed in \cite{bjelakovic_quantum_2008, bjelakovic_classical_2007,  bjelakovic_arbitrarily_2013}, as they consider an unknown channel taken from a known set that is applied consistently, while we consider a new channel that draws randomly from a set of known channels at each use, but which reveals the result of this random drawing at the output.

In addition to the classical capacity, we will also discuss the \textit{potential Holevo capacity} introduced in \cite{winter_potential_2016}. The \textit{potential Holevo capacity} is the maximum Holevo information gain of a channel when assisted with another channel,
\begin{align*}
\chi^{(pot)}(\Phi) = \max_{\Psi} ( \chi(\Phi \otimes \Psi) - \chi(\Psi) ) \pl.
\end{align*}
It is clear that $\chi^{(pot)}(\Phi)\geq C(\Phi) \geq \chi(\Phi)$, and we say that the channel's Holevo information is \emph{additive} if $C(\Phi) = \chi(\Phi)$; \emph{strongly additive} if $\chi^{(pot)}(\Phi) = \chi(\Phi)$. In both cases its classical capacity $C(\Phi)$ is fully characterized by the single-letter Holevo information $\chi(\Phi)$.
We refer to the book by Wilde \cite{wilde} for more information about quantum channel capacity and other basics in quantum Shannon theory.
\subsection{Squashed Entanglement}The squashed entanglement of a bipartite state $\rho^{AB}$, defined in \cite{christandl_squashed_2004}, is
\begin{align*}
E_{sq}(A,B)_{\rho} = \inf \{ \frac{1}{2} I(A;B|C)_\rho \pl | \pl tr_{C}(\rho^{ABC}) =\rho^{AB}  \}\pl,
\end{align*}
where the infimum runs over all extensions $\rho^{ABC}$ of $\rho^{AB}$. We will use the following properties of the squashed entanglement in this paper:
\begin{enumerate}
\item[i)] Convexity: let $\rho^{AB}= \sum_x p(x) \rho_x^{AB}$ be a convex combination of states $\{\rho_x\}$, then
\begin{align}E_{sq}(A,B)_{\rho} \leq \sum_x p_x E_{sq}(A,B)_{\rho^x}\pl.\label{convex}\end{align}
\item[ii)] Monogamy: let $\rho^{BB_1\cdots B_k}$ be a $(k+1)$-partite state, then
\begin{align}\sum_{j=1}^{k} E_{sq}(B, B_j)_\rho \leq S(B)_\rho \pl.\label{eq:sqprops}\end{align}
\item[iii)] $1$-norm faithfulness:
\begin{align}\min \Big \{\| \rho^{AB} - \si^{AB} \|_1 \pl |\pl  \si =\sum_x p_x \si^{A}_x \otimes \si^{B}_x \pl \text{separable}\Big \} \leq 3.1 |A| \sqrt[4]{E_{sq}(A,B)_\rho}\label{faithful}\end{align}
where $\norm{\cdot}{1}$ is the trace class norm.
\end{enumerate}
The convexity and monogamy of squashed entanglement are well-established in the literature \cite{christandl_squashed_2004}. The faithfulness of Schatten norms are proved by Brand\~{a}o et al \cite{brandao_faithful_2011} and Li and Winter \cite{li_squashed_2014}. The readers are also referred to the survey paper \cite{horodecki_quantum_2009} for more information about measures of quantum entanglement.
\subsection{Continuity of entropic expressions}
Given a bipartite state $\rho^{AB}$, its conditional entropy (conditional on $B$) is given by $S(A|B)_\rho=S(AB)_\rho-S(B)_\rho$.
The continuity of conditional expectation is characterized by the Alicki-Fannes inequality as introduced in \cite{alicki_continuity_2004} and later refined by Winter in \cite{winter_tight_2016}. For two bipartite states $\rho$ and $\si$ and $\delta\in [0,1]$,
\begin{equation}
\begin{split}
  \frac{1}{2}\|\rho^{AB} - \sigma^{AB} \|_1 \leq \delta \implies
 |S(A|B)_\rho - S(A|B)_\sigma | \leq  2\delta \log |A| + (1+\delta) h(\frac{\delta}{1+\delta}) \pl,
\end{split}
\label{eq:alickifannes}
\end{equation}
where $h(p) = - p \log p - (1 - p) \log (1 - p)$ is the binary entropy function, $|A|$ is the dimension of $A$. Note that for $0\le \delta\le 1$,
\[(1+\delta) h(\frac{\delta}{1+\delta})\le -2(1+\delta)\frac{\delta}{(1+\delta)}\log \frac{\delta}{(1+\delta)} \le 2\delta\log \frac{1+\delta}{\delta} \pl.\]
We will use the following (weaker but simpler) variant of \eqref{eq:alickifannes},
\begin{align}
 |S(A|B)_\rho - S(A|B)_\sigma | \leq  2\delta \log \frac{2|A|}{\delta} \pl.\label{AF}
\end{align}

\section{Heralded Averaging}
\label{sec:heravg}
For a channel $\Phi:A\to B$ and a bipartite state $\rho \in {B_0\ten A}$, we introduce the following notation of squashed entanglement,
\[E_{sq}(B_0, \Phi)_\rho := E_{sq}(B_0, B)_{id\ten \Phi(\rho)} \pl.\]
and similarly for conditional entropy and mutual information
\[S(B_0|\Phi)_\rho := S(B_0| B)_{id\ten \Phi(\rho)}\pl, \pl I(B_0;\Phi)_\rho := I(B_0; B)_{id\ten \Phi(\rho)} \pl.\]
\begin{lemma}
\label{thm:entdil}Let $Z_{k}(\bigphi)$ be a heralded channel. For any state $\rho \in B_0 \otimes A_{1} \cdots A_{n}$,
\begin{align*}
E_{sq}(B_0, Z_{k}(\bigphi))_{\rho} \leq \frac{1}{L}S(B_0)_\rho \pl,
\end{align*}
where $L = \floor*{n/k}$, the largest integer less or equal than $n/k$.
\end{lemma}
\begin{proof} Let $S_n$ be the symmetric group of the integer set $[n]=\{1,\cdots,n\}$. The heralded channel $Z_{k}(\bigphi)$ can be rewritten as
\[Z_{k}(\bigphi) (\rho) =\frac{1}{n!} \sum_{\si\in S_n} \Big (\Phi^{\si(R)} \otimes \canchan^{\si(R)^c}\Big ) (\rho) \otimes \ketbra{\si(R)}^Y\pl,\]
where $R$ can be any $k$-subset of $[n]$ and the summation runs over all permutations $\si$.
Then we find $L$ disjoint subset $R_1,R_2,\cdots, R_L$ with each $|R_l|=k$, and have
\[Z_{k}(\bigphi) (\rho) =\frac{1}{L}\sum_{l=1}^L \frac{1}{n!} \sum_{\si\in S_n} \Big (\Phi^{\si(R_l)} \otimes \canchan^{\si(R_l)^c}\Big ) (\rho) \otimes \ketbra{\si(R_l)}\pl. \]
Exchanging the summation, we have
\begin{align*}
E_{sq}(B_0, Z_{k}(\bigphi))_{\rho} =&
	E_{sq}\Big (B_0,\frac{1}{n!} \sum_{\si\in S_n} \frac{1}{L}\sum_{l=1}^L (\Phi^{\si(R_l)} \otimes \canchan^{\si(R_l)^c} ) \otimes \ketbra{\si(R_l)}^Y\Big )_{\rho}
\\\leq & \frac{1}{n!} \sum_{\si\in S_n} \frac{1}{L}\sum_{l=1}^L E_{sq}\Big(B_0,  (\Phi^{\si(R_l)} \otimes \canchan^{\si(R_l)^c})\Big)_{\rho}
\\= & \frac{1}{n!} \sum_{\si\in S_n} \frac{1}{L}\sum_{l=1}^L E_{sq}\Big(B_0,  \Phi^{\si(R_l)} \Big)_{\rho}
\\\leq & \frac{1}{n!} \sum_{\si\in S_n} \frac{1}{L}\sum_{l=1}^L E_{sq}\Big(B_0,  A^{\si(R_l)} \Big)_{\rho}\pl.
\end{align*}
Here the first inequality follows from convexity and the fact $\ketbra{\si(R_l)}^Y$ is a classical signal. The second equality is because $\Theta$'s are trivial channels. The last inequality is the data processing inequality of squashed entanglement (squashed entanglement is not increasing under local operation). Note that for any permutation $\si$, $\si(R_1), \cdots,\si(R_L)$ are disjoint positions because $R_1,R_2,\cdots, R_L$ are. Thus for any $\si$,
\[\sum_{l=1}^L E_{sq}(B_0,  A^{\si(R_l)} )_{\rho}\le E_{sq}(B_0,  A^{\si(R_1\cdots R_l)} )_{\rho}\le S(B_0)_{\rho} \pl ,\]
which completes the proof.
\end{proof}
The above argument is easily adapted to tensor products of heralded channels. Let $\bigphi_i=\{\Phi_{i,1},\Phi_{i,2},\cdots,\Phi_{i,n_i}\}, 1\leq i\leq m$ be $m$ classes of quantum channels such that each class consists of $n_i$ quantum channels. Let $A=\ten_{i=1}^m (\ten_{j=1}^{n_i}A_{i,j})$ be the total input system and $B=\ten_{i=1}^m(\ten_{j=1}^{n_i} B_{i,j})$ be the quantum part of output system. Then the tensor product of heralded channels
\[\ten_{i=1}^m Z_{k_i}(\bigphi_i) (\rho) =\frac{1}{\prod \binom{n_i}{k_i}}\sum_{R=R_1\cdots R_m , |R_i|=k_i}(\Phi^{R}\ten \Theta^{R^c}) (\rho) \ten \ketbra{R}^Y\pl\]
is from $A$ to $BY$,
where the heralding signal $R$ now is an ensemble of the heralding signals $R_i\subset [n_i]$ for each $Z_{k_i}(\bigphi_i)$. Now consider $\si_i$ to be a permutation from $S_{n_i}$ and denote $\si=(\si_1,\cdots,\si_m)$ as the corresponding element in the product group $S_{n_1}\times \cdots\times S_{n_m}$. We have
\[\ten_{i=1}^m Z_{k_i}(\bigphi_i) (\rho) =\frac{1}{\prod n_i!}\sum_{\si}(\Phi^{\si(R)}\ten \Theta^{\si(R)^c}) (\rho) \ten \ketbra{\si(R)}^Y\pl,\]
where $\si(R)=\si_1(R_1)\cdots\si_m(\R_m)$ is the ensemble of shifted positions and the summation runs over all permutations $\si\in S_{n_1}\times \cdots\times S_{n_m}$.
\begin{cor}\label{cor:entdil}
Let $Z_{k_i}(\bigphi_i), 1\le i\le m$ be a family of heralded channels. For any state $\rho\in B_0\ten A$, where $A= \otimes_{i=1}^m (A_{1,i} \cdots A_{n_i,i})$, we have
\begin{align*}
E_{sq}(B_0, \ten_{i=1}^m Z_{k_i}(\bigphi_i))_{\rho} \leq \frac{1}{L}S(B_0)_\rho \pl,
\end{align*}
where $\displaystyle L = \min_i \floor*{n_i/k_i}$.
\end{cor}
\begin{rem}{\rm
We may replace the squashed entanglement above by any information measure obeying convexity \eqref{convex} and monogamy \eqref{eq:sqprops}. In the future, it is possible that other faithful, monogamous entanglement measures may be discovered, yielding superior bounds via the same techniques as in this paper.}
\end{rem}
\begin{theorem}
\label{thm:entblockperm}
Let $\ten_{i=1}^m Z_{k_i}(\bigphi_i): A \to B$ be the tensor channel of a family of heralded channels. Let $B_0$ be an extra system with dimension $|B_0|$. Suppose $\flambda= 1/\floor*{ \min_i  n_i/k_i} $ is small enough such that $3.1|B_0|\sqrt[4]{\flambda \log|B_0|}\le 2$. Then for any state $\rho\in B_0\ten A$, there exists a state $\eta \in B_0\ten B$ that is separable between $B_0$ and $B$ such that
\begin{align*}
| S(B_0 | \otimes_{i=1}^m Z_{k_i}(\bigphi_i))_{\rho} - S(B_0 | B)_{\eta} |
	 \leq 3.1 |B_0| \sqrt[4]{\flambda S(B_0)_{\rho}}\log \bigg(\frac{1}{3.1 \sqrt[4]{\flambda S(B_0)_{\rho}}}\bigg)\pl.
\end{align*}
\end{theorem}
\begin{proof}
By Corollary \ref{cor:entdil},
\[E_{sq}(B_0, \otimes_{i=1}^m Z_{k_i}(\bigphi_i))_\rho \leq \flambda S(B_0)_\rho\pl.\]
Then we apply the faithfulness of squashed entanglement \eqref{faithful} to find a separable $\eta$ such that
\[\norm{id_{B_0}\otimes_{i=1}^m Z_{k_i}(\bigphi_i)(\rho)-\eta}{1}\le 3.1 |B_0| \sqrt[4]{\flambda S(B_0)_{\rho}}\le 2\pl.\]
Therefore, by Alicki-Fannes inequality \eqref{AF},
\begin{align*}
| S(B_0 | \otimes_{i=1}^m Z_{k_i} (\bigphi_i))_{\rho} - S(B_0|B)_{\eta} |
	 \leq 3.1 |B_0| \sqrt[4]{\flambda S(B_0)_{\rho}} \log \frac{4}{3.1 \sqrt[4]{\flambda S(B_0)_{\rho}}} &\pl. \qedhere
\end{align*}
\end{proof}

\begin{section}{Iterative Bounds on Holevo Information}
\label{sec:it}
 We first state the technical version of our theorems regarding Holevo information in Section \ref{sec:holevothms} and leave the proofs of the main Theorem \ref{thm:holevostrongadd} and its corollaries to Section \ref{sec:holevoproofs}.
\subsection{Holevo information estimates for heralded channels}
\label{sec:holevothms}
Our main technical theorem considers tensor products of heralded channels or flagged switch channels, showing that heralding removes the superadditivity from a channel by diluting the available entanglement. It is the building block we use to prove the simpler corollaries, and ultimately Theorem \ref{thm:introer} and \ref{thm:introher}. In the following, we always assume that $\bigphi_i=\{\Phi_{i,1},\cdots,\Phi_{i,n_i}\},  \bigpsi_i=\{\Psi_{i,1},\cdots,\Psi_{i,n_i}\}$ are pairs of families of channel such that for each $(i,j)$, $\Phi_{i,j},\Psi_{i,j}:A_{i,j} \to B_{i,j}$ share the same input and output system.
\begin{theorem}
\label{thm:holevostrongadd}Let $Z_{k_i}(\bigphi_i; \bigpsi_i), 1\le i\le m$ be a family of flagged switch channels and $\Phi_0: A_0 \to B_0$ be an arbitrary channel. Suppose that
${\flambda = 1/\lfloor \min_i  n_i/k_i  \rfloor}$ is small enough such that $3.1|B_0|\sqrt[4]{\flambda \log|B_0|}\le 2$. Then
\begin{align*}
\chi(\Phi_0 \otimes_{i=1}^m Z_{k_i} (\bigphi_i; \bigpsi_i)) \leq \chi(\Phi_0) + &\chi(\otimes_{i=1}^m Z_{k_i} (\bigphi_i))
    +\sum_{i=1}^m (1-\frac{ k_i}{n_i}) \sum_{j=1}^{n_i} \chi^{(pot)}(\Psi_{i,j})
	\\ &+ 6.2 |B_0| \sqrt[4]{\flambda \log |B_0|}\log \bigg (\frac{4}{3.1 \sqrt[4]{\flambda \log |B_0|}} \bigg )\pl.
\end{align*}
\end{theorem}
The above theorem bounds the additivity violation between heralded channels, and an arbitrary extra channel $\Phi_0$. The form is analogous to an approximate strong additivity for heralded channels with small $\flambda$, but it depends on the dimension of output system $B_0$. So it is not truly a strong additivity, as the bound becomes trivial when $\Phi_0$ is a channel with a sufficiently large output space. Based off of Theorem \ref{thm:holevostrongadd}, we prove that heralded channels are approximately additive.
\begin{cor}
\label{cor:holevoadd}
Let the output system $B_{i,j}$ to each $\Phi_{i,j}$ be of dimension at most $d$ and let
${\flambda = 1/\lfloor \min_i  n_i/k_i  \rfloor}$ be small enough such that $3.1d\sqrt[4]{\flambda \log d}\le 2$. Then,
\begin{equation}
\begin{split}
 \chi(\otimes_{i=1}^m Z_{k_i} (\bigphi_i; \bigpsi_i))
  \leq
	\sum_{i=1}^m \frac{k_i}{n_i} \sum_{j=1}^{n_i} \chi(\Phi_{i,j}) + \sum_{i=1}^m (1-\frac{ k_i}{n_i}) \sum_{j=1}^{k_i} \chi^{(pot)}(\Psi_{i,j}) \\
	+ 6.2 d(\sum_{i}k_i) \sqrt[4]{\flambda \log d}\log \bigg (\frac{4}{3.1 \sqrt[4]{\flambda \log d}} \bigg )
\nonumber
\end{split}
\end{equation}
As a consequence, for a single flagged switch channel $Z_k(\bigphi ; \bigpsi)$
\begin{align*}
C(Z_k(\bigphi ; \bigpsi)) \leq \frac{k}{n} \sum_{j=1}^{n} \chi(\Phi_j) + (1-\frac{  k}{n}) \sum_{j=1}^{n} \chi^{(pot)}(\Psi_j)
	+ 6.2 dk \sqrt[4]{\flambda \log d}\log \bigg (\frac{4}{3.1 \sqrt[4]{\flambda \log d}} \bigg )\pl.
\end{align*}
\end{cor}
Unlike Theorem \ref{thm:holevostrongadd}, this corollary depends only on the dimension of the constituent channels $\Phi_{i,j}$. Therefore, it is an approximate additivity when $\flambda$ is small enough. It implies Theorem \ref{thm:introher} by applying to heralded channels (for which $\Psi$ are trivial channels).

Finally, we derive a bound analogous to strong additivity entirely in terms of single-letter expressions.
\begin{cor}
\label{cor:holevofinal}
Let the output system to each $\Phi_{i,j}$ and $\Phi_0$ be of dimension at most $d$ and let
${\flambda = 1/\lfloor \min_i  n_i/k_i  \rfloor}$ be small enough such that $3.1d\sqrt[4]{\flambda \log d}\le 2$.
\begin{align*}
\chi(\Phi_0 \ten& \otimes_{i=1}^m Z_{k_i} (\bigphi_i; \bigpsi_i))  \leq \chi(\Phi_0) + \sum_{i=1}^m \frac{k_i}{n_i} \sum_{j=1}^{n_i} \chi(\Phi_{i,j}) +
	\sum_{i=1}^m (1-\frac{k_i}{n_i}) \sum_{j=1}^{n_i} \chi^{(pot)}(\Psi_{i,j}) \\
	&+ 6.2 d (1+\sum_{i}k_i)(\flambda \log d)^{\frac{1}{4}}\log \bigg (\frac{d}{3.1 (\flambda \log d)^{\frac{1}{4}}} \bigg )\pl.
\nonumber	
\end{align*}
\end{cor}

This corollary combines the Theorem \ref{thm:holevostrongadd} and Corollary \ref{cor:holevoadd} to bound all superadditivity present, allowing us to write the Holevo information of a product of heralded channels enhanced by an arbitrary $\Phi_0$ as a sum of single-letter expressions and correction terms. We say that this bound is merely analogous to strong additivity due to the dimension-dependence on the channel $\Phi_0$, which prevents us from bounding the potential capacity of a heralded channel.

\begin{subsection}{Now the proofs}
\label{sec:holevoproofs}
We start with a lemma showing that Holevo information is not increasing when conditioning with a separable correlation.
\begin{lemma}
\label{lem:holevopart}
Let $\Phi_0:A_0 \to B_0$ be a quantum channel. Let $B$ be any quantum system, and $\rho_x^{A_0B}$ be a family of separable bipartite state. Then for any $\eta = \sum_x p_x\eta_x$, where $\eta_x=id\ten \Phi_0(\rho_x)$ and $\{p_x\}$ probability distribution, we have
\begin{align*}
S(B_0 | B)_{\eta} - \sum_x p_x S(B_0 | B)_{\eta_x} \leq \chi(\Phi_0)\pl.
\end{align*}
\end{lemma}
\begin{proof} By separability, for each $x$, we may write
\begin{align*}
\eta_x = \sum_{j} p_{x,j} \eta_{x,j}^{B_0} \otimes \eta_{x,j}^{B} \pl
\end{align*}
as a convex combination of product states. Define a classical to quantum channel  $\Phi^{cq}: X \rightarrow B$, where $X$ is an extra classical system, by
\begin{align*}
\Phi^{cq} ( \ket{x,j}\bra{x,j}) = \eta_{x,j}^{B}\pl.
\end{align*}
From this, we define the classical-quantum states
\begin{align*}
\eta'_x = \sum_{j} p_{x,j}\eta_{x,j}^{B_0} \otimes \ket{x,j} \bra{x,j}
\end{align*}
and observe that $\eta_x = \Phi^{cq}(\eta_x')$ for each $x$.
Applying the data processing inequality for conditional entropy,
\begin{align*}
S(B_0 | B)_{\eta} - \sum_x p_x S(B_0 | B)_{\eta_x} &\le S(B_0)_{\eta} - \sum_x p_x S(B_0 | B)_{\eta_x}
\\ &\le S(B_0)_{\eta} - \sum_x p_x S(B_0 | X)_{\eta_x'}
\\&\le S(B_0)_{\eta} - \sum_x p_xp_{x,j} S(B_0 )_{\eta_{x,j}}\pl.
\end{align*}
This has the form of the Holevo information for the state
$\sum_{x,j}p_xp_{x,j}\ketbra{x,j}\ten\eta_{x,j}$. Therefore, it is less than the Holevo information of $\Phi_0$.
\end{proof}

The next lemma replaces each $\Psi$ channel by its potential capacity in an upper bound, which reduces the discussion of Theorem \ref{thm:holevostrongadd} to heralded channels.
\begin{lemma}
\label{lem:potrep}
Let $Z_{k_i}(\bigphi_i; \bigpsi_i), 1\le i\le m$ be a family of flagged switch channels and let $\Phi_0$ be an arbitrary channel $\Phi_0: A_0 \to B_0$. Then
\begin{align*}
\chi(\Phi_0 \otimes_{i=1}^m Z_{k_i}(\bigphi_i; \bigpsi_i))  \leq
 \chi(\Phi_0 \otimes_{i=1}^m Z_{k_i}(\bigphi_i)) + \sum_{i=1}^m (1-\frac{k_i}{n_i}) \sum_{j=1}^{k_i} \chi^{(pot)}(\Psi_{i,j}) \, .
\end{align*}
\end{lemma}
\begin{proof} Let $B=B_0\ten (\ten_{i=1}^{m}\ten_{j=1}^{n_i}B_{j,i})$ be the full quantum output system. For an classical quantum input state $\rho^{XA}=\sum_{x}p_x\ketbra{x}\ten \rho_x^A$, the output state is
\begin{align*}
\omega^{XBY}&=\sum_{x}p_x\ketbra{x}\ten (\ten_{i=1}^m Z_{k_i}(\bigphi_i;\bigpsi_i)(\rho_x))\\&=\sum_x \sum_{R=R_1\cdots R_m , |R_i|=k_i} \frac{1}{\prod \binom{n_i}{k_i}}p_x\ketbra{x}^X\ten(\Phi^{R}\ten \Psi^{R^c})(\rho_x) \ten \ketbra{R}^Y\pl.
\end{align*}
Note that the marginal distributions (reduced density) on the two classical system $XY$ is independent. Thus we have
\begin{align*}
\chi(\Phi_0  \otimes_{i=1}^m Z_{k_i} (\bigphi_i ; \bigpsi_i))
	= \sup_{\rho^{XA}}I(X:BY)_\omega= \sup_{\rho^{XA}}\frac{1}{\prod \binom{n_i}{k_i}}\sum_{R} I(X:B)_{\omega(R)}\pl,
\end{align*}
where
\[\omega(R) =\sum_x p_x\ketbra{x}^X\ten(\Phi^{R}\ten \Psi^{R^c})(\rho_x)\]
is the outcome of the heralding signal $R$. In each $I(X:B)_{\omega(R)}$, we could then replace $R^c$ systems by the potential capacities of $\Psi$ channels without decreasing the expression. That is
\[I(X:B)_{\omega(R)}\le I(X:B)_{\omega(R)'}+\sum_{j\in R_i}\chi^{(pot)}(\Psi_{i,j}) \pl,\]
where
\[\omega(R)'=\sum_x p_x\ketbra{x}^X\ten(\Phi^{R}\ten \Theta^{R^c})(\rho_x)\]
is the corresponding output of heralded channels. The result follows from summing over all $R$ with uniform probabilities.
\end{proof}

\begin{proof}[{\bf Proof of Theorem \ref{thm:holevostrongadd}}]
First, by Lemma \ref{lem:potrep} it is sufficient to estimate $\chi(\Phi_0  \otimes_{i=1}^m Z_{k_i} (\bigphi_i))$. Note that for a classical-quantum state $\omega^{XB_0B}=\sum_{x}p_x\ketbra{x}\ten \omega_x^{B_0B}$
\begin{align*}
 I(X ; B_0 B )_\omega &= S(B_0B)_\omega- \sum_x p_xS(B_0B)_{\omega_x}\\
&=
\big(S(B_0|B)_\omega- \sum_x p_xS(B_0|B)_{\omega_x}\big)+ \big(S(B)_\omega- \sum_x p_xS(B)_{\omega_x}\big) \pl.
\end{align*}
The second half is majorized by the Holevo information of the tensor product of heralded channels if \[\omega^{XB_0B}=id_X\ten \Phi_0\otimes_{i=1}^m Z_{k_i} (\bigphi_i) (\rho^{XA_0A})\]
as an output state. For the first half, we use the factorization property
\[ \Phi_0 \otimes_{i=1}^m Z_{k_i}(\bigphi_i))=\big( \Phi_0 \ten (\ten_{i,j} \Phi_{i,j})\ten id_Y \big) \circ (id_{A_0} \ten_{i=1}^m Z_{k_i}^{n_i}(id:\Theta)) \pl,\]
where ``$id$'' represents the identity channel and  $Z_{k_i}^{n_i}(id:\Theta)$ is the erasure channel. For each $\rho_x$, there exists state $\eta_x^{A_0AY}$ separable between $A_0$ and $AY$ such that
\[ \norm{id \ten_{i=1}^m Z_{k_i}^{n_i}(id:\Theta)(\rho_x)- \eta_x}{1}\le 3.1 |B_0| \sqrt[4]{\flambda S(B_0)_{\rho_x}}\pl.\]
Denote $\delta=3.1 |B_0| \sqrt[4]{\flambda \log|B_0|}$. Thus,
\begin{align*}
&|S(\Phi_0 | \otimes_{i=1}^m Z_{k_i} (\bigphi_i))_{\rho_x}-S(\Phi_0 | (\ten_{i,j} \Phi_{i,j})\ten id_Y )_{\eta_x}|\le \delta\log \frac{4|B_0|}{\delta} \pl, \\
&|S(\Phi_0 | \otimes_{i=1}^m Z_{k_i} (\bigphi_i))_{\rho}-S(\Phi_0 | (\ten_{i,j} \Phi_{i,j})\ten id_Y )_{\eta}|\le \delta\log \frac{4|B_0|}{\delta}\pl.
\end{align*}
Therefore, by Lemma \ref{lem:holevopart} and the triangle inequality
\begin{align*}
&S(B_0|B)_\omega- \sum_x p_xS(B_0|B)_{\omega_x}\\\le &S(\Phi_0|(\ten_{i,j} \Phi_{i,j})\ten id_Y )_\eta- \sum_x p_xS(\Phi_0|(\ten_{i,j} \Phi_{i,j})\ten id_Y )_{\eta_x}+ 2 \delta\log \frac{4|B_0|}{\delta}  \\ \le &\chi(\Phi_0)+2 \delta\log \frac{4|B_0|}{\delta} \pl.
\end{align*}
Putting these together, we obtain
\begin{align*}
 I(X ; B_0 B )_\omega &=
\big(S(B_0|B)_\omega- \sum_x p_xS(B_0|B)_{\omega_x}\big)+ \big(S(B)_\omega- \sum_x p_xS(B)_{\omega_x}\big)
\\ &\le \chi(\Phi_0)+
\chi(\otimes_{i=1}^m Z_{k_i} (\bigphi_i)) + 2 \delta\log \frac{4|B_0|}{\delta}. 	\qedhere
\end{align*}
\end{proof}
\begin{proof}[{\bf Proof of Corollary \ref{cor:holevoadd}}]
As before, we need only estimate the Holevo information of $\ten_{i=1}^m Z_{k_i}(\bigphi_i)$. We start by the averaging on the $j$th position being a position in the set $R$ for a heralded channel $Z_k(\Phi_1,\cdots,\Phi_n)$. Since each $R$ contains $k$ positions, up to a re-ordering based on the classical signal,
\begin{align*}
Z_k(\Phi_1, ..., \Phi_n ) (\rho) &= \frac{1}{k} \sum_{j=1}^n \frac{1}{\binom{n}{k}} \sum_{ R} \Big (\Phi^{R} \otimes \Theta^{R^c} \Big ) (\rho) \otimes \ketbra{R} \\&=
\frac{1}{n} \sum_{j=1}^n \frac{1}{\binom{n-1}{k-1}} \sum_{j\in R} \Big (\Phi^{R} \otimes \Theta^{R^c} \Big ) (\rho) \otimes \ketbra{R}
\\ &= \frac{1}{n} \sum_{j=1}^n \Phi_j\ten \Big(\frac{1}{\binom{n-1}{k-1}} \sum_{j\notin R, |R|=k-1} \Phi^{R} \otimes \Theta^{R^c} (\rho) \otimes \ketbra{R} \Big )
\\ &= \frac{1}{n} \sum_{j=1}^n \Phi_j\ten Z_{k-1}(\Phi_1,\cdots,\Phi_{j-1},\Phi_{j+1}, \cdots,\Phi_n) (\rho)
\pl.
\end{align*}
Denote $\delta=3.1 |B_0| \sqrt[4]{\flambda \log|B_0|}$. We use the convexity of Holevo information (\cite{gao_capacity_2016}, Proposition 4.3) and Theorem \ref{thm:holevostrongadd},
\begin{align*}
 \chi(\otimes_{i=1}^m Z_{n_i} (\bigphi_i))  = &\chi \bigg(\frac{1}{n_1} \sum_{j}  \Phi_{1,j} \otimes Z_{k_1-1}(\Phi_{1,1},\cdots,\Phi_{1,j-1},\Phi_{1,j+1}, \cdots,\Phi_{n_1,1})  \ten	\otimes_{i=2}^m Z_{n_i} (\bigphi_i)\bigg) \\
	\leq &\frac{1}{n_1} \sum_{j} \chi \bigg ( \Phi_{1,j} \otimes Z_{k_1-1}(\Phi_{1,1},\cdots,\Phi_{1,j-1},\Phi_{1,j+1}, \cdots,\Phi_{n_1,1})  \ten	\otimes_{i=2}^m Z_{k_i} (\bigphi_i) \bigg ) \\
\leq &\frac{1}{n_1} \sum_{j} \chi(\Phi_{1,j}) +2\delta\log \frac{4d}{\delta} \\&+ \chi ( Z_{n-1}(\Phi_{1,1}, ..., \Phi_{1,j-1}, \Phi_{1,j+1}, ..., \Phi_{1,k_1})\ten \otimes_{i=2}^m Z_{k_i} (\bigphi_i) )\big)
\end{align*}
We may repeat this procedure to separate out all $\Phi$ positions in each $Z_{k_i} (\bigphi_i)$. As we replace each $\Phi_{i,j}$ by its Holevo information plus the correction term, we are reducing $n$ and $k$ by the same amount, so $\flambda=1/\floor{\min k_i/n_i}$ does not increase. Thus $2\delta\log \frac{4d}{\delta}$ is a uniform bound for the correction term at all steps.
Therefore
\begin{align*}
\chi(\otimes_{i=1}^m Z_{k_i} (\bigphi_i)) \le \sum_{i=1}^m \frac{k_i}{n_i} \sum_{j=1}^{n_i} \chi(\Phi_{i,j}) + 2(\sum_i k_i)\delta\log \frac{4d}{\delta} \pl.
\end{align*}
Because the correction term does not depends on the number of tensors $m$, regularizing the expression for $Z_{k} (\bigphi)^{\ten m}$ yields the bound for classical capacity.
\end{proof}

\end{subsection}
\end{section}

\begin{section}{Erasure Channels Composed with Other Channels}
\label{sec:related}
We define a \emph{generalized erasure channel} as an erasure channel composed with some channel $\Phi$ in its successful trial,
\begin{align*}
\label{eq:erase}
Z_\lambda(\Phi)(\rho) = \lambda \Phi(\rho)\otimes\ket{0}\bra{0}   + (1 - \lambda) \si\otimes \ket{1}\bra{1}\pl,
\end{align*}
where $\lambda \in [0,1]$ and $\si$ is some fixed state. We call the successes those cases in which $\Phi$ was applied with a classical signal of $\ket{0}\bra{0}$. $Z_\lambda(\Phi)$ can be viewed as a heralded channel in terms of success probability, rather than imposing a definite number of successes. Indeed, tensor products of the erasure channel can be expressed as a probabilistic sum over heralded channels,
\begin{align*}
\big ( Z_\lambda(\Phi_1) \otimes \cdots \otimes Z_\lambda(\Phi_n) \big ) (\rho)&=\sum_{k=0}^n\lambda^k (1 - \lambda)^{n-k}\sum_{|R|=k} (\Phi^R\ten \Theta^{R^c}) (\rho) \ten \ketbra{R}
\\&= \sum_{k=0}^n \binom{n}{k} \lambda^k (1 - \lambda)^{n - k} Z_k(\Phi_1, \cdots, \Phi_n) (\rho) \pl.
\end{align*}
This is easy to see by noting that each erasure channel is a binary heralding with probability $\lambda$, so the overall distribution of successes is binomial. In the following we denote
\[Z_k^n(\Phi)=Z_k(\underbrace{\Phi,\Phi,\cdots,\Phi}_{n})\]
as the heralded failure channel with $n$ identical $\Phi$'s and fixed success number $k$.
\begin{theorem}
\label{thm:eraserherald}
For any $n$ and $\la$,
\begin{align*}
|\chi(Z_\lambda(\Phi)^{\otimes n}) - \chi(Z_{\lfloor\lambda n\rfloor}^n(\Phi) )| \leq
	\big(1+\sqrt{n\la(1-\la)}\big)\chi^{(pot)}(\Phi)\pl.
\end{align*}
In particular, the classical capacity of the erasure channel $Z_\lambda(\Phi)$ can be rewritten by
\begin{align*}
C(Z_\lambda(\Phi)) = \lim_{n \rightarrow \infty} \frac{1}{n} \chi(Z_{\lfloor\lambda n\rfloor}^n(\Phi))\pl.
\end{align*}
\end{theorem}
We prepare our proof with a lemma.
\begin{lemma}
\label{lem:potcapbound}Let $k_1\le k_2$. For any classical-quantum input $\rho$,
\begin{align*}
\label{eq:potcapbound}
0\le I(X:Z^n_{k_1} (\Phi))_\rho - I(X:  Z^n_{k_2} (\Phi) )_\rho \leq (k_2 - k_1) \chi^{(pot)}(\Phi)\pl.
\end{align*}
\end{lemma}
\begin{proof}
 By the argument in Lemma \ref{lem:holevopart},
\[I(X:Z^n_{k} (\Phi))_\rho=\sum_{k=0}^n \binom{n}{k} I(X:\Phi^R\ten \Theta^{R^c})_\rho = \frac{1}{\binom{n}{k}}\sum_{|R|=k}I(X:\Phi^R)_\rho\pl.\]
Therefore, we have
\begin{align*}
I(X:Z^n_{k_2} (\Phi))_\rho&=\frac{1}{\binom{n}{k_2}}\sum_{|R|=k_2}I(X:\Phi^R)_\rho
\\&=\frac{1}{\binom{n}{k_2}}\sum_{|R|=k_2}\frac{1}{\binom{k_2}{k_1}}\sum_{|P|=k_1,P\subset R}I(X:\Phi^P\ten\Phi^{R/P})_\rho
\\&\le\frac{1}{\binom{n}{k_2}}\sum_{|R|=k_2}\frac{1}{\binom{k_2}{k_1}}\sum_{|P|=k_1,P\subset R}\big(I(X:\Phi^P)_\rho+(k_2-k_1)\chi^{(pot)}(\Phi)\big)
\\&=\frac{1}{\binom{n}{k_1}}\sum_{|P|=k_1}I(X:\Phi^P)_\rho+(k_2-k_1)\chi^{(pot)}(\Phi)
\\&=I(X:Z^n_{k_1} (\Phi))_\rho+(k_2-k_1)\chi^{(pot)}(\Phi)\pl.
\end{align*}
The last step is because each $k_1$-subset $P$ has been counted $\binom{n-k_1}{k_2-k_1}$ times as a subset of some $k_2$-set $R$.
The inequality follows similarly.
\end{proof}
\begin{proof}[{\bf Proof of Theorem \ref{thm:eraserherald}.}]
Let $A$ be the input system of $\Phi$.
Then the two channels $Z_\la(\Phi)^{\ten n}$ and $Z_{\lfloor\la n\rfloor}^n(\Phi)$ have the same input $A^{\ten n}$. Denote $\rho^{XA^n}$ by a classical-quantum state.
By the triangle inequality and Lemma \ref{lem:potcapbound},
\begin{align*}
|\chi(Z_\la(\Phi)^{\ten n})-\chi(Z_{\lfloor\la n\rfloor}^n(\Phi))|&=|\sup_\rho
\sum_{k=0}^n \binom{n}{k} \lambda^k (1 - \lambda)^{ n-k}I(X:Z_k^n(\Phi))-\sup_\rho I(X:Z_{\lfloor\la m\rfloor}^n(\Phi))_\rho| \\
&\le \sup_\rho
|\sum_{k=0}^n \binom{n}{k} \lambda^k (1 - \lambda)^{n-k}I(X:Z_k^n(\Phi)) - I(X:Z_{\lfloor\la m\rfloor}^n(\Phi))_\rho|\\
&\le
\sum_{k=0}^n \binom{n}{k} \lambda^k (1 - \lambda)^{n-k}|\sup_\rho I(X:Z_k^n(\Phi))- I(X:Z_{\lfloor\la m\rfloor}^n(\Phi))_\rho|\\
&\le
\chi^{pot}(\Phi)\sum_{k=0}^n \binom{n}{k} \lambda^k (1 - \lambda)^{n-k}|k-\lfloor \la n\rfloor|\\
&\le
\chi^{pot}(\Phi)\big(1+\sum_{k=0}^n \binom{n}{k} \lambda^k (1 - \lambda)^{n-k}|k- \la n| \big)
\pl.
\end{align*}
Recall that the variance of binomial distribution is $n\la(1-\la)$.
We obtain by H{\"\o}lder inequality that
\begin{align*}&\sum_{k=0}^n \binom{n}{k} \lambda^k (1 - \lambda)^{n-k}|k- \la n|\le (\sum_{k=0}^n \binom{n}{k} \lambda^k (1 - \lambda)^{n-k}|k- \la n|^2)^{1/2}=\sqrt{n\la(1-\la)}\pl.  \qedhere
\end{align*}
\end{proof}
\begin{cor}
\label{cor:erasercap}
Let $d$ be the dimension of the output system for $\Phi$. Then
\begin{align*}
C(Z_\lambda(\Phi)) \leq \lambda \Bigg ( \chi(\Phi) + 6.2 d \sqrt[4]{\lambda \log d}\log \bigg ( \frac{d}{3.1 \sqrt[4]{\lambda \log d}} \bigg )  \Bigg )\pl.
\end{align*}
\end{cor}

The above corollary \ref{cor:erasercap} follows from Theorem \ref{thm:eraserherald} and Corollary \ref{cor:holevoadd}. It is the technical version of Theorem \ref{thm:introer}, and confirms our conjecture for classical capacity. Based on this, we propose the following definition.

\begin{definition}
We define the \textit{post-selected classical capacity} of a channel $\Phi$ with success rate $\lambda$ as following,
\begin{align*}
C^{PS}_\lambda(\Phi) := \frac{1}{\lambda} C(Z_\lambda(\Phi))= \lim_{n \rightarrow \infty} \frac{1}{\lambda n} C(Z_\lambda(\Phi)^{\otimes n}) \pl.
\end{align*}
\end{definition}
The post-selected capacity is motivated operationally by the common experimental scenario: we perform an experiment many times with some success rate $\lambda$, but we only wish to consider the instances in which the experiment was run successfully in our results. This is equivalent to projecting the total outcome onto the heralding signal $\ket{0}\bra{0}^Y$ subspace. Its significance is that even though we discard the failures, they still matter for superadditive quantities such as classical capacity.

\end{section}
\begin{section}{Other Applications}
\label{sec:other}
\subsection{Capacities with small entangled blocksize coding}
While we would like to immediately extend our results to quantum and private capacities, these quantities have a more complicated interaction with heralding and erasure. The quantum erasure channel, for instance, has zero quantum capacity if the success probability is less than $1/2$ \cite{bennett_capacities_1997}, subsuming the values for which our near-additivity results would apply. It is possible that the quantum capacity has a different monogamy phenomenon from Holevo information, but that is beyond the scope of this paper. We might instead consider the quantum and private capacities with feedback, but these are much more complicated to estimate with entropy techniques \cite{leung_capacity_2009}.

In practice, the realities of quantum hardware often limit the size of entangled blocks. When the success probability of a heralding process is small compared with the entangled blocksize of inputs, most successes will be the only success in their respective block. In this case, we show that superadditivity is limited by a simpler argument. For example, recent proposals for trapped ion computing schemes mention heralded photonic Bell state analysis as a method for connecting trapped ion arrays, a necessity due to the number-limiting interactions of trapped ions in the same array \cite{brown_co-designing_2016}. The array size restricts the entangled input blocksize for inter-array communication, so lossy optical links may exhibit superadditivity-limiting effects even with 2-way classical assistance. The generation of scalably entangled ``cat" states is at the time of this writing an ongoing effort in experimental quantum information. Realistic calculations of attainable quantum and private rates with today's hardware probably should assume small entangled blocksizes.
\begin{prop}
 Let $\Phi_1,\cdots, \Phi_n$ be a family of quantum channels, and let $\lambda \ll 1/n$. Let $F^{(1)}$ be a function mapping densities to positive real numbers, and define its output on a quantum channel by $F^{(1)}(\Phi) = \max_\rho\{F^{(1)} (\Phi(\rho)) \}$. Assume $F^{(1)}$ is superadditive on quantum channels, additive for separable input states, convex in the input state and channel, and admit a well-defined expression of the form $F^{(pot)}(\Phi) = \max_\Psi \{F^{(1)}(\Phi \otimes \Psi) - F^{(1)}(\Psi)\}$ such that $F^{(pot)}(\canchan) = 0$ if $\canchan$ is a trivial channel for which the output provides no information about the input (see \cite{wilde} and \cite{winter_potential_2016} for more information on different capacities). Then
\begin{align*}
 F^{(1)}(Z_\lambda(\Phi_1) \otimes ... \otimes Z_\lambda(\Phi_n)) -  \lambda \sum_{i=1}^n F^{(1)}(\Phi_i) \leq O \bigg (\lambda^2 \sum_i \Big ( F^{(pot)}(\Phi_i) - F^{(1)}(\Phi_i) \Big ) \bigg )\pl.
\end{align*}
\label{thm:mblock}
\end{prop}
\begin{proof}If only one success has occurred at position $i$, then
\begin{equation}
\label{eq:pulloneoff}
F^{(1)} (\Theta_1 ... \Phi_i ... \Theta_n) = F^{(1)}(\Phi_i),
\end{equation}
since trivial channels contribute nothing. For some tensor product $\Phi^R\ten \Theta^{R^c}$, define $\Psi_1 ... \Psi_n$ such that $\Phi^R\ten \Theta^{R^c} = \Psi_1 \ten ... \ten \Psi_n$.
\begin{equation}
\label{eq:correc}
\begin{split}
& F^{(1)}(\Psi_1 \ten ... \ten \Psi_n) \leq F^{(pot)}(\Psi_1) + F^{(1)}(\Psi_2 \ten ... \ten \Psi_n) \\ 
& = F^{(1)}(\Psi_1) + F^{(1)}(\Psi_2 \ten ... \ten \Psi_n) + F^{(pot)}(\Psi_1) - F^{(1)}(\Psi_1) \\
... & \leq \sum_{i \in R} F^{(1)}(\Phi_1) + \sum_{i \in R : i \neq |R|} \Big ( F^{(pot)}(\Phi_1) - F^{(1)}(\Phi_1) \Big ),
\end{split}
\end{equation}
where the last inequaltiy is obtained by iterating the first and discarding the $F^{(pot)} (\Theta_i)$ terms that are zero anyway.  Note that the final correction term has one fewer channel correction term than are channels, because equation \ref{eq:pulloneoff} shows that we do not need to add a correction term when only one success is involved. Using equations \ref{eq:correc} and \ref{eq:pulloneoff}, and the convexity of $F^{(1)}$ in the channel,
\begin{align*}& F^{(1)}(Z_\lambda(\Phi_1) \otimes ... \otimes Z_\lambda(\Phi_n)) \le \sum_{k=0}^n\sum_{|R|=k}\la^{k}(1-\la)^{n-k}
F^{(1)}(\Phi^R\ten \Theta^{R^c})
\\  \le &\sum_{k=1}^n\sum_{|R|=k}\la^{k}(1-\la)^{n-k}
\sum_{i\in R}F^{(1)}(\Phi_i)+\sum_{k=2}^n\sum_{|R|=k}\la^{k}(1-\la)^{n-k}
\sum_{i\in R : i \neq |R|}(F^{(pot)}(\Phi_i)-F^{(1)}(\Phi_i))
\\  \le &\sum_{k=1}^n\binom{n}{k}\la^{k}(1-\la)^{n-k}\frac{k}{n}
(\sum_{i=1}^n F^{(1)}(\Phi_i))+\sum_{k=2}^n\binom{n}{k}\la^{k}(1-\la)^{n-k}\frac{k - 1}{n}
\sum_{i=1}^n (F^{(pot)}(\Phi_i)-F^{(1)}(\Phi_i))
\\  \le &\la\sum_{i=1}^n F^{(1)}(\Phi_i)
+\la \sum_{k=2}^n \binom{n-1}{k-1}\la^{k-1}(1-\la)^{n-k}
(\sum_{i=1}^n F^{(pot)}(\Phi_i)-F^{(1)}(\Phi_i))
\\  \le &\la\sum_{i=1}^n F^{(1)}(\Phi_i)
+\la(1-(1-\la)^{n-1})
\sum_{i=1}^n(F^{(pot)}(\Phi_i)-F^{(1)}(\Phi_i))
\\  \le &\la\sum_{i=1}^n F^{(1)}(\Phi_i)
+O\big(\lambda^2(
\sum_{i=1}^n F^{(pot)}(\Phi_i) - F^{(1)}(\Phi_i))\big)\pl.
\end{align*}
\end{proof}
\begin{rem}{\rm The intuition behind proposition \ref{thm:mblock} is that when the erasure probability of an eraser channel is sufficiently high compared to the maximum entangled blocksize of the encoder, it is unlikely that more than one non-trivial channel will appear in any block. This result is not based on entanglement monogamy.}
\end{rem}
\begin{rem}{\rm Proposition \ref{thm:mblock} applies to the block-limited coherent information $Q^{(m)}$, but when $\lambda \ll 1$, the unassisted quantum capacity vanishes anyway due to antidegradability of the erasure channel \cite{bennett_capacities_1997}. Proposition \ref{thm:mblock} applies to the Holevo information $\chi$ as well as the classical capacity with limited entanglement assistance as disussed in \cite{junge_channel_2015}. We conjecture that it will apply to the quantum capacity with classical feedback and 2-way classical communication as discussed in \cite{leung_capacity_2009}, but we cannot yet confirm the necessary properties to apply our results, as we do not have a one-shot entropy expression for these quantities.}
\end{rem}

\subsection{Entanglement monogamy in quantum games}
\label{sec:games}
Going beyond the case of communication, the idea that entanglement monogamy implies bounds on entanglement-based super-additivity also applies to non-local games. A two-player game $G=(A,B,X,Y,\pi,v)$ is played between a referee and two isolated players, Alice and Bob, who communicate only with the referee and not between themselves. The referee chooses a question pair $(x, y)$ according to some probability distribution $\pi$ on the question alphabets $X\times Y$, sending $x$ to Alice and $y$ to Bob. The two players respond with answers $a$ and $b$ respectively from answer sets $A$ and $B$. They win the game if $v(x, y, a, b) = 1$ for the verification function $v: X\times Y\times A\times B \to \{0,1\}$ and lose otherwise. The classical value of the game \[val(G)= \sup_{a_x,b_y } \sum_{x,y,a,b}\pi(x,y) v(a, b, x, y)   \int_\Omega a_x(\omega)b_y(\omega)d \mathbb{P}(\omega)\]
is the maximum winning probability when Alice and Bob are allowed to use optimal deterministic strategies $\sum_{a}a_x(\omega)=\sum_{b}b_y(\omega)=1$ based on some classical correlation (common randomness) $\mathbb{P}(\omega)$. The entangled or quantum value of the game \[val^*(G)= \sup_{\rho,E_x^a,F_{y}^b} \sum_{x,y,a,b}\pi(x,y)  v(a, b, x, y) tr(\rho E_x^a\ten F_y^b)\] allows Alice and Bob to answer the question by performing POVMs  $\sum_{a}E_x^a=1, \sum_{b}F_y^b=1$ with some auxilliary bipartite entangled state $\rho$. It is clear that for all games, $val(G)\le val^*(G)$. A bound on the value with classical states, such as Bell's inequality, can be evidence of quantum non-locality (see \cite{palazuelos_survey_2016} for more information on non-local games).

We consider the application of entanglement monogamy in the following scenario.
Suppose the player Alice has a single system $A$, which can share classical correlation or quantum entanglement with a large number of ``Bob'' players $B_1,\cdots,B_n$ simultaneously. The referee randomly (with equal probability) selects a ``Bob'' player $B_i$ and plays the corresponding bipartite game $G_i=(A,B_i,X_i,Y_i,\pi_i,v_i)$ with Alice and $B_i$. This is analogous to the heralding process, because although the two players $A$ and $B_i$ know the game $G_i$ after $B_i$ is selected by the referee, the strategy and auxilliary resource, classical or quantum, has to be prepared between Alice and $B_1,\cdots,B_n$ before the game is played. In this situation, the non-classicality in entangled games will be bounded by entanglement monogamy if the size of Alice's system is fixed, because the average amount of entanglement Alice's system can have with each $B_i$ is limited by the number of ``Bob'' players.

For $\{G_i\}_{1\le i\le n}$, we define the average entangled value when Alice having an at most $d$-dimensional quantum system as follows,
\begin{align*}
Aval_d^*(\{G_i\})=\sup_{\rho,E_x^a, F_{1,y}^b,\cdots, F_{n,y}^b} \frac{1}{n}\sum_{i=1}^n\sum_{a,b,x,y}\pi_i(x,y) v_i(a, b, x, y)tr(\rho^{AB_i} E_x^a\ten F_{i,y}^b)\pl,
\end{align*}
where for each $x$ and $y$, $E_x^a, F_{1,y}^b,\cdots, F_{n,y}^b$ are POVMs on $A, B_1,\cdots, B_n$ respectively and $\rho^{AB_1\cdots B_n}$ is a multipartite state with $|A|$ is at most $d$. Also, the average classical value is given by
\begin{align*}
&Aval(\{G_i\})= \frac{1}{n}\sum_{i=1}^n val(G_i)\pl.
\end{align*}
because the classical correlation used for different $G_i$ can be combined.

\begin{theorem}
Let $G_1,\cdots, G_n$ be a family of bipartite games that Alice play respectively with the players $B_1, \cdots,B_n$. Then
\begin{align*}
 Aval_d^*(\{G_i\}) -Aval(\{G_i\})\leq n^{-1/4}d(\log d)^{1/4}\pl.
\end{align*}
\end{theorem}
\begin{proof}Let $G=(A,B,X,Y,\pi,v)$ be a bipartite game. For fixed axillary systems $\A,\B$ and POVMs $E_x^a,F_y^b$, the value function can be viewed as a positive linear functional $l_G$ on the trace-class $S_1(\A\ten \B)$,
\[l_G(\rho^{\A\B})=\sum_{x,y,a,b}\pi(x,y)  v(a, b, x, y) tr(\rho E_x^a\ten F_y^b)\pl.\]
Note that $l_G$ is of norm at most $1$. Then for a separable $\sigma$ and an arbitrary $\rho$,
\begin{align*}
l_G(\rho) \leq l_G(\rho - \sigma)  + l_G(\sigma)
\leq  \| \rho - \sigma \|_1 + val(G) \pl.
\end{align*}
Now suppose that the axillary quantum system $\A$ of Alice is of dimension at most $d$. We know by the monogamy of entanglement \eqref{eq:sqprops} that for any $\rho^{AB_1\cdots B_n}$,
\begin{align*}
\label{eq:gamemon}
\frac{1}{n} \sum_{i=1}^n E_{sq}(A, B_i)_\rho \leq \frac{1}{n} S(A)_\rho \leq \frac{\log d}{n}\pl.
\end{align*}
It follows from the faithfulness of squashed entanglement that there exists a state $\sigma^{A B_i}$ separable on $A$ and $B_i$ such that
\begin{align*}
\| \rho^{A B_i} - \sigma^{A B_i} \|_1 \leq 3.1 d \sqrt[4]{E_{sq}(A,B_i)_\rho}
\end{align*}
Thus,
\begin{align*}
l_{G_i}(\rho^{AB_i})\leq 3.1 d \sqrt[4]{E_{sq}(A,B_i)_\rho} + val(G_i) \pl,
\end{align*}
and then the average entangled value obeys
\begin{align*}
Aval_d^*(\{G_i\}) \leq \frac{1}{n}\sum_{i=1}^n\big(3.1 d \sqrt[4]{E_{sq}(A,B_i)_\rho} + val(G_i)\big)
	\leq Aval(\{G_i\}) + 3.1 d\frac{1}{n}\sum_{i=1}^n\sqrt[4]{E_{sq}(A,B_i)_\rho}\pl.
\end{align*}
By H{\"o}lder's inequality, we have
\begin{align*}
\sum_{i=1}^n\sqrt[4]{E_{sq}(A,B_i)_\rho}\le n^{3/4}(\sum_{i=1}^n E_{sq}(A,B_i)_\rho)^{1/4}\le n^{3/4}(\log d)^{1/4}\pl.
\end{align*}
Therefore,
\begin{align*}
Aval_d^*(\{G_i\}) -Aval(\{G_i\})\leq n^{-1/4}d(\log d)^{1/4}\pl.
\end{align*}
\end{proof}
We can see that with large $n$ and fixed $d$, the classicality violation decays to $0$. In keeping with the results of \cite{junge_large_2011}, however, increasing the dimension of Alice and Bobs' systems simultaneously with the number of measurements can compensate for a decline in entanglement. For this reason, we do not expect to a bound that is independent of the dimension of Alice's system.

This simple application in quantum games may give some intuition behind the phenomenon in Theorem \ref{thm:mblock}, and possibly those in Section \ref{sec:holevothms}. The existence of a monogamous, faithful entanglement measure (like squashed entanglement) necessarily constrains entanglement-dependent quantities to have a sort of monogamy of their own. This technique can be applied to any situation in which one seeks entanglement-dependent effects in the presence of a heralding-like process.

\end{section}
\begin{section}{Discussion and Conclusion}
\label{sec:conclusion}
It follows intuitively that the monogamy of a faithful entanglement measure would also imply monogamy for entanglement-dependent quantities. Recent advances in the theory of squashed entanglement \cite{li_squashed_2014} have combined these properties in a sufficiently quantitative way as to directly derive bounds on entropy expressions. As a property of entangled states, we expect superadditivity of channel capacity to be monogamous. If many entangled systems enhance the capacity of a single channel, monogamy would intuitively imply that the enhancement due to each is correspondingly limited. While this alone might be an interesting anecdote and of potential interest to routines designed to optimize quantum input states, the heralded channel provides a direct operational application of the monogamy of superadditivity. When an arbitrary channel is randomly distributed amongst a larger number of copies of a strongly additive channel, entanglement in the input state goes to waste, and the combined channel loses superadditivity.

One obvious application is to situations encountered in experimental physics, such as photonics, where the non-trivial channel created by decoherence combines with the erasure channel arising from photon loss. Due to the difficulty of creating and maintaining large-scale entanglement in transmitted photons, and the likelihood that entangling inputs with eventually destroyed copies sacrifices some capacity, many realistic quantum information systems will probably transmit in the nearly-additive realm of Theorem \ref{thm:mblock} and not gain much by attempting to exploit superadditivity.

Finally, we note that this technique is not limited to capacity or entropy expressions. Any quantity that depends on entanglement and is faithful with the trace norm may show a monogamy-like effect due to comparison with squashed entanglement. This broader implication of entanglement monogamy bounding entanglement-dependent quantities may also be significant.
\end{section}
\begin{section}{Acknowledgements}
We thank Mark Wilde and Debbie Leung for helpful conversations. We thank Paul Kwiat and the Kwiat lab at UIUC for conversations and experiences that gave us a sense of the importance of heralded channels in practical applications.
\end{section}

\bibliography{monogamy}
\bibliographystyle{unsrt}
\end{document}